\documentclass[11pt]{article}
\usepackage[utf8]{inputenc}
\usepackage{fullpage}
\usepackage{amsfonts}
\usepackage{amssymb}
\usepackage{tikz}
\usetikzlibrary{arrows.meta}
\usepackage[dvipsnames]{xcolor}
\usepackage{mathtools}

\usepackage{amsthm}
\usepackage{thm-restate}

\usepackage{color}
\usepackage{hyperref}
\hypersetup{
    colorlinks=true,
    linkcolor=blue,
    citecolor=blue
    }

\usepackage[nameinlink, noabbrev, capitalize]{cleveref}

\theoremstyle{plain}

\newtheorem{thm}{Theorem}[section]
\newtheorem{theorem}[thm]{Theorem}

\newtheorem{conjecture}[thm]{Conjecture}

\newtheorem{corollary}[thm]{Corollary}

\newtheorem{question}[thm]{Question}
\newtheorem{claim}[thm]{Claim}
\newtheorem*{claim*}{Claim}

\theoremstyle{remark}
\newtheorem{remark}[thm]{Remark}

\def\r{{\mathrm{r}}}                               
\def\p{{\mathrm{p}}}                               
\def\nnr{{\mathrm{nnr}}}                              
\DeclareMathOperator{\ur}{srr}
\DeclareMathOperator{\tr}{tr}

\DeclareMathOperator{\sumc}{sum}

\newcommand{\eps}{\varepsilon}

\newcommand{\R}{\mathbb{R}}

\DeclareMathOperator{\polylog}{polylog}

\title{The Log-Rank Conjecture: New Equivalent Formulations}

\author{Lianna Hambardzumyan \thanks{The University of Copenhagen, Denmark. \texttt{lianna.hambardzumyan@gmail.com}. This work is funded by the European Research Council (ERC) under grant agreement no. 101125652 (ALBA).} \footnotemark[4]
\and
Shachar Lovett \thanks{University of California San Diego, CA, USA. \texttt{slovett@ucsd.edu}. Research supported by Simons Investigator Award \#929894 and NSF award CCF-2425349.}
\and
Morgan Shirley \thanks{Lund University, Sweden. \texttt{morgan.shirley@cs.lth.se}. Supported by Knut and Alice Wallenberg grant KAW 2023.0116.} \thanks{Most of the work was done while the first and third authors were postdoctoral researchers at University of Victoria, Canada, and funded by NSERC.}
}

\begin{document}
	
\maketitle

\begin{abstract}
	The log-rank conjecture is a longstanding open problem with multiple equivalent formulations in complexity theory and mathematics. In its linear-algebraic form, it asserts that the rank and partitioning number of a Boolean matrix are quasi-polynomially related.

	We propose a relaxed but still equivalent version of the conjecture based on a new matrix parameter, signed rectangle rank: the minimum number of all-1 rectangles needed to express the Boolean matrix as a $\pm 1$-sum. Signed rectangle rank lies between rank and partition number, and our main result shows that it is in fact equivalent to rank up to a logarithmic factor. Additionally, we extend the main result to tensors.
	This reframes the log-rank conjecture as: can every signed decomposition of a Boolean matrix be made positive with only quasi-polynomial blowup? 
	
	As an application, we prove an equivalence between the log-rank conjecture and a conjecture of Lovett and Singer–Sudan on cross-intersecting set systems. 
\end{abstract}

\section{Introduction}
The log-rank conjecture is one of the most well-known problems in complexity theory, and despite extensive work it remains unsolved. It asserts that for a Boolean matrix its communication complexity and the logarithm of its matrix rank over the reals are polynomially related. An equivalent linear-algebraic formulation of the conjecture is that for Boolean matrices, the matrix rank and  partitioning number (sometimes called the binary rank) are quasi-polynomially related. 

More formally, for a Boolean matrix $M$ let $\r(M)$ denote its rank over the reals, and let $\p(M)$ denote its partitioning number: the minimum number of all-$1$ submatrices that partition the $1$-entries of $M$.  Equivalently, the partitioning number is the minimum number $p$ such that $M = \sum_{i=1}^p R_i$, where each $R_i$ is a \emph{primitive} matrix---an all-$1$ submatrix, possibly after adding all-zero rows and columns. 
If we relax the decomposition to allow general rank-1 matrices instead of primitive matrices, we recover the standard matrix rank. That is, $\r(M)$ is the minimum number $r$ such that $M=\sum_{i=1}^r M_i$, with each $M_i$ of rank-$1$. In this sense, matrix rank can be seen as a relaxation of the partitioning number, and trivially $\r(M) \leq \p(M)$. 
The log-rank conjecture asks how well this relaxation estimates the partitioning number:
\begin{conjecture}[Log-rank conjecture \cite{lovasz1988lattices}]\label{conj:log-rank}
	For any Boolean matrix $M$, \[\log \p(M) \leq (\log \r(M))^{O(1)}.\]
\end{conjecture}

The log-rank conjecture was first posed by Lov\'asz and Saks \cite{lovasz1988lattices} in the context of communication complexity: is the deterministic communication complexity of a Boolean matrix polynomially related to the logarithm of its rank over the reals? A closely related question had appeared even earlier in graph theory. There, the log-rank conjecture is equivalent to asking whether the logarithm of a graph's chromatic number is quasi-polynomially related to the rank of its adjacency matrix \cite{van1976bound,fajtlowicz1988conjectures,lovasz1988lattices}.

Regarding the state of the art, the best known upper bound is due to Sudakov and Tomon \cite{tomon2024matrix}, improving on previous work of Lovett \cite{lovett2016communication}, and shows that
\[ \log \p(M) \leq O\left(\sqrt{\r(M)} \right).\]
The largest known separation follows from \cite[Corollary 3]{Balodis21}, which yields a matrix $M$ satisfying $\log \p(M) \geq \Omega(\log^2 \r(M))$. 
For a more detailed overview of the log-rank conjecture and its equivalent formulations, we refer the reader to the survey by Lee and Shraibman \cite{lee2023around}, as well as the textbooks by Jukna \cite{jukna2012boolean} and Rao and Yehudayoff \cite{rao2020communication}.

In this paper, we consider a ``gradual relaxation'' from partitioning number to rank via an intermediate complexity measure: the \emph{signed rectangle rank}. This notion allows decomposition of the matrix into \textit{signed} primitive matrices. Formally, let the \emph{signed rectangle rank} of $M$, denoted by $\ur(A)$, be the minimum number $t$ such that $M= \sum_{i=1}^t \eps_i R_i$, where $\eps_i \in \{1,-1\}$ and each $R_i$ is a primitive matrix. Trivially, we have \[\r(M) \leq \ur(M) \leq \p(M).\]
A promising approach to the log-rank conjecture is to study how ``close'' signed rectangle rank is to rank and partitioning number. To resolve the conjecture, it would suffice to either prove both of the following statements or disprove one of them:
\begin{itemize}
	\item $\ur(M)$ is quasi-polynomially related to $\r(M)$, and
	\item $\ur(M)$ is quasi-polynomially relate to $\p(M)$.
\end{itemize}

The main result of the paper is that rank and signed rectangle rank are tightly related.

\begin{theorem}[Main Theorem]\label{thm:main}
	Every Boolean matrix $M$ of rank $r$ can be written as a $\pm 1$-linear-combination of at most $O(r \log r)$ primitive matrices, that is,
	\[\ur(M) \leq O(r \log r).\]
\end{theorem}

From this we immediately get an equivalent formulation of the log-rank conjecture:
\begin{conjecture}
	\label{conj:unary_vs_rank}
	For every Boolean matrix $M$, $\log \p(M) \leq (\log \ur(M))^{O(1)}$.
\end{conjecture}
\begin{corollary}[from \cref{thm:main}] \label{corollary: main}
	The log-rank conjecture (\cref{conj:log-rank}) is equivalent to \cref{conj:unary_vs_rank}.
\end{corollary}

A natural question left open here is whether the bound in \cref{thm:main} can be improved, or if it is already tight:

\begin{restatable}{question}{QuestionTightness} \label{question: tightness}
	Is it true that for every Boolean matrix $M$, $\ur(M) \leq O(\r(M))$?
\end{restatable}

A result of Lindström~\cite{LinCombinatorial1965}  shows that the bound is tight for our proof technique (see the discussion at the end of \cref{sec:main-theorem}). Therefore, any improvements on the bound require different ideas. 

\paragraph*{Power of cancellations vs.~monotonicity.} The log-rank conjecture compares two ways of measuring the complexity of a matrix: rank and partitioning number. These measures differ in two fundamental ways:

\begin{itemize}
	\item \textbf{Combinatorial structure}: Rank expresses a matrix as a summation of rank-1 matrices with arbitrary real coefficients and no apparent combinatorial structure. In contrast, partitioning number only allows primitive matrices, which are purely \emph{combinatorial} objects.
	\item \textbf{Cancellations}: Rank allows for \emph{cancellations} in the summation. Partitioning, on the other hand,  is a \emph{monotone} model: all-1 submatrices (also known as \emph{rectangles}) must be combined positively and disjointly.
\end{itemize}
The log-rank conjecture asks whether these two measures are nevertheless quasi-polynomially related.

The measure we study, \emph{signed rectangle rank}, is an intermediate notion that allows cancellations, but only at the level of rectangles. In this way, it relaxes partitioning by permitting signed combinations of rectangles, while still restricting to simple combinatorial building blocks. Intuitively, our main theorem shows that allowing cancellation between rectangles already captures essentially all of the linear structure present in a Boolean matrix.

An analogous perspective arises from studying another measure of a matrix, the \emph{non-negative rank}. If one forbids cancellations but allows non-negative real values, without imposing any combinatorial structure on the decomposition, the resulting measure is the non-negative rank. Formally, non-negative rank of a $n \times n$ matrix $M$, denoted by $\nnr(M)$, is the minimum number $r$ such that $M=\sum_{i=1}^{r} u_i v_i^T$, where $u_i,v_i \in \R^n$ are non-negative column vectors. Trivially, we have
\[ \r(M) \leq \nnr(M) \leq \p(M). \]

Analogous to our main theorem, it is known that for Boolean matrices this relaxation of partitioning number does not increase its power too much (this follows from a result of Lovász and Saks~\cite{lovaszCommunication1993}; see e.g.\ Watson~\cite{WatNonnegative2016}):
\[ \log \p(M) \leq O(\log^2 \nnr(M)). \]

Consequently, the log rank conjecture is equivalent to asking whether there is a quasi-polynomial relationship between rank and non-negative rank, namely, is $\log \nnr(M) \leq (\log \r(M))^{O(1)}$? Comparing this with \cref{corollary: main}, our results imply that the log-rank conjecture also reduces to converting a signed rectangle decomposition into a fully monotone one, with at most a quasi-polynomial increase in the number of primitive matrices. See \cref{fig: measures} for an illustration of the relationships between these measures.

Together, these paint a picture of what is left to understand about the log-rank conjecture. The fact that rank allows arbitrary real coefficients and partitioning number is combinatorial is a red herring: the only relevant question is whether allowing cancellations is significantly stronger than monotonicity.

\begin{figure}[ht]
	\centering
	\newcommand{\rectmargin}{0.8}
	\newcommand{\bendamount}{25}
	\newcommand{\dashedepsilon}{0.45}
	\begin{tikzpicture}[scale=1.2]
		\draw[thin, fill=blue!20!white, opacity=0.2]
		(-2, 2+\rectmargin) --
		(0+\rectmargin, 0) --
		(0, 0-\rectmargin) -- node[below, sloped, opacity=.7] {\small{monotone measures}}
		(-2-\rectmargin, 2) --
		cycle;
		
		\draw[thin, fill=blue!20!white, opacity=0.2]
		(2, 2+\rectmargin) --
		(0-\rectmargin, 0) --
		(0, 0-\rectmargin) -- node[below, sloped, opacity=0.7] {\small{combinatorial measures}}
		(2+\rectmargin, 2) --
		cycle;
		
		\node (rank) at (0, 4) {\large $\r$};
		\node (nnrank) at (-2, 2) {\large $\nnr$};
		\node (srr) at (2, 2) {\large $\ur$};
		\node (par) at (0, 0) {\large $\p$};
		
		\path[-{Latex}, black!70!white]
		(par) edge[bend right=\bendamount] (srr)
		(par) edge[bend left=\bendamount] (nnrank)
		(srr) edge[bend right=\bendamount] (rank)
		(nnrank) edge[bend left=\bendamount] (rank);
		
		\path[-{Latex}, very thick, black!70!white]
		(nnrank) edge[bend left=\bendamount] node[below, sloped] {\footnotesize{\cite{lovaszCommunication1993}}} (par);
		
		\path[-{Latex}, very thick, NavyBlue]
		(rank) edge[bend right=\bendamount] node[above, sloped] {\footnotesize{Thm 1.2}} (srr);
		
		\draw[dashed, thick, red]
		(-1-\dashedepsilon,3+\dashedepsilon) -- (1+\dashedepsilon,1-\dashedepsilon);
	\end{tikzpicture}
	\caption{Relationships between the matrix measures discussed; an arrow $a \rightarrow b$ indicates that $\log b \leq \log^{O(1)} a$. The thin arrows indicate bounds which are trivial from the definitions. The red dashed line indicates that proving the log-rank conjecture is equivalent to showing that \emph{any} measure below the line is quasi-polynomially related to \emph{any} measure above the line.} \label{fig: measures}
\end{figure}

\paragraph*{Equivalent conjecture on cross-intersecting set systems.} 

Let $\mathcal{S}, \mathcal{T} \subseteq 2^{[d]}$ be two set families. 
We say the pair $(\mathcal{S}, \mathcal{T})$ is \emph{$L$-cross-intersecting} for some $L \subseteq \{0,\ldots,d\}$ if for all $S \in \mathcal{S}$ and $T \in \mathcal{T}$, we have $|S \cap T| \in L$; that is, the size of every pairwise intersection belongs to $L$. Cross-intersecting set families have been widely studied in combinatorics, with much of the work focusing on their extremal properties \cite{frankl1987forbidden, sgall1999, keevash2005set, hunter2024disjoint}. 

The following conjecture about $\{a,b\}$-cross-intersecting set systems was independently proposed by Lovett~\cite{lovett2021} and Singer–Sudan~\cite{singer2022}, who both observed that it is implied by the log-rank conjecture.
\begin{conjecture}\label{conj:set_systems}
	Let $\mathcal{S} = \{S_1, \ldots, S_m\}$ and $\mathcal{T} = \{T_1, \ldots, T_n\}$ be an $\{a,b\}$-cross-intersecting pair of families from $2^{[d]}$, where $a,b \in \{0, \ldots, d\}$. Then there exist subfamilies $\mathcal{A} \subseteq \mathcal{S}$ and $\mathcal{B} \subseteq \mathcal{T}$ such that $(\mathcal{A}, \mathcal{B})$ is either $\{a\}$- or $\{b\}$-cross-intersecting, and
	\[
	|\mathcal{A}||\mathcal{B}| \geq 2^{-\polylog(d)} \cdot |\mathcal{S}||\mathcal{T}|.
	\]
\end{conjecture}

As an application of \cref{thm:main}, we show that the log-rank conjecture is equivalent to this conjecture.\footnote{This equivalence was claimed in \cite{singer2022} as a parenthetical remark. However, this was a typo. The remark was meant to claim the implication from the log-rank conjecture \cite{madhu_personal_communication}.}

\begin{restatable}{theorem}{setsystems}\label{thm:set_systems}
	The log-rank conjecture (\cref{conj:log-rank}) is equivalent to \cref{conj:set_systems}.
\end{restatable}

\paragraph*{Tensors.} In addition to the matrix case we extend \cref{thm:main} to Boolean tensors, showing an analogous relation between the tensor rank and signed rectangle rank of a tensor. This has implications for the generalization of the log-rank conjecture to \emph{number-in-hand} communication complexity (asked explicitly in \cite{DKWPartition2011}).

\paragraph*{Organization.} We give the proof of \cref{thm:main} in \cref{sec:main-theorem} and discuss the open question of how tight the bound is.  In \cref{sec:set_systems}, we prove the equivalence of the log-rank conjecture to \cref{conj:set_systems}. Finally, we prove the extension of the main theorem to tensors in \cref{sec:tensors}.

\section{Proof of the Main Theorem} \label{sec:main-theorem}
Let $M$ be an $m \times n$ Boolean matrix of rank $r$.
For a subset $S$ of the columns, define its \emph{column-sum} to be the column vector $c$ such that for $i \in [m]$, $c(i) = \sum_{j \in S} M(i, j)$. 
That is, $c$ is the entrywise sum of all the columns of $S$ along all the rows. Denote the column-sum of $S$ as $\sumc(S)$. Call a subset $S$ of columns \emph{independent} if all the subsets of $S$ have distinct column-sums.

\begin{claim} \label{cl:IS-size}
	Let $S$ be an independent set of columns of $M$. Then $|S| \leq O(r \log r)$, where $r$ is the rank of $M$.
\end{claim}
\begin{proof}
	Let $A_S$ be the matrix whose columns are the column-sums of $S$, so $A_{S}$ has $2^{|S|}$ columns. The entries of $A_S$ are in $\{0,\ldots,|S|\}$ as they are sums of at most $|S|$ entries of $M$, which take value $\{0, 1\}$. The rank of $A_S$ is at most $r$ because its columns are in the span of the columns of $M$, which itself has rank $r$. Thus, there is a set of rows $R$ with $|R| = r$ such that, for any column $c$ of $A_S$, the values $(c(i))_{i \in R}$ are sufficient to determine every entry of $c$. Therefore, there are at most $(|S|+1)^r$  unique columns of $A_S$, and since every column of $A_S$ is unique, we have that the subsets of $S$ have at most $(|S|+1)^r$ column-sums.  Therefore, $2^{|S|} \leq (|S|+1)^r$, which after taking logarithms gives $|S| \leq O(r \log |S|)$, which implies $|S| \leq O(r \log r)$.
\end{proof}

\begin{claim}\label{cl:linear_combin}
	Let $S$ be a maximal independent set of columns of $M$.  Then every column of $M$ can be expressed as a  $\pm 1$-linear combination of columns in $S$.
\end{claim}

\begin{proof}
	This trivially holds for the columns in $S$. 
	Fix a column $c \notin S$. Since $S$ is a maximal independent set, $S \cup \{c\}$ is not independent. This means that there are two subsets $A$ and $B$ of $S \cup \{c\}$ with the same column-sum. We can assume that these subsets are disjoint; if they are not, then removing the common columns still results in equal column-sums. Note $A$ and $B$ cannot both exclude $c$, as this would contradict $S$ being independent. Assume $c \in B$, and let $B' = B\setminus \{c\}$. Combining this with $\sumc(A)=\sumc(B)$, we get $\sumc(A) = \sumc(B') + c$. Noting that $A \cap B' = \emptyset$ and $A,B' \subseteq S$ concludes that $c=\sumc(A)-\sumc(B')$ is the desired linear combination.
\end{proof}

Combining these two claims concludes the proof of \cref{thm:main} as follows. Let $S$ be the largest independent set of columns of $M$. By \cref{cl:linear_combin} every column $y$ can be expressed as a $\pm 1$-linear combination of  columns in $S$. Thus, $M$ can be written as

\[ M(x, y) = \sum_{c \in S} \alpha_{c}(y) c(x), \]
%\[ M(x, y) = \sum_{c \in S} \alpha_{c}(y)^{+} c(x) - \sum_{c \in S} \alpha_{c}(y)^{-} c(x), \]
where each coefficient $\alpha_{c}(y)\in \{-1, 0, 1\}$. Decompose $\alpha_{c}(y) = \alpha_{c}^{+}(y) - \alpha_{c}^{-}(y)$, where $\alpha_{c}^{+}(y), \alpha_{c}^{-}(y) \in \{0,1\}$.
For each column $c \in S$, define
\[ R_c^{+}(x, y) = \alpha_{c}^{+}(y) c(x) \qquad \mbox{and} \qquad R_c^{-}(x, y) = \alpha_{c}^{-}(y) c(x) \]

Each of $R^{+}_c$ and $R^{-}_c$ is either an all-zeroes matrix or a primitive matrix, since they are outer products of $\{0, 1\}$-valued vectors. Hence, $M = \sum_{c \in S} (R_c^{+}- R_c^{-})$, which expresses $M$ as $\pm 1$-sum of at most $2|S|$ primitive matrices. Finally, applying the bound on $|S|$ from \cref{cl:IS-size}, we conclude that $\ur(M) \leq O(r \log r)$.

We remind the reader of the following open question:

\QuestionTightness*

An approach one might hope to take is to improve the statement of \cref{cl:IS-size}, that is, to prove a better upper bound on the size of an independent set of columns in a low-rank Boolean matrix. However, a result of Lindström~\cite[Theorem 1]{LinCombinatorial1965} shows that \cref{cl:IS-size} is tight. Concretely, Lindström gives a construction of a set of $\Theta(r \log r)$ vectors in $\{0,1\}^r$ which are independent (that is, all their column sums are distinct).

%; therefore, any rank-$r$ matrix containing $M^{\star}$ as a submatrix will have an independent set of columns of size $\Theta(r \log r)$.

\section{Equivalence to the cross-intersecting set systems conjecture}\label{sec:set_systems}
In order to prove \cref{thm:set_systems}, we will use the equivalent version of the log-rank conjecture by Nisan and Wigderson \cite{NW}. A submatrix of a matrix is called a \emph{monochromatic rectangle} if all of its entries have the same value.
\begin{conjecture}[\cite{NW}]\label{conj:nw}
	Any Boolean matrix $M$ has a monochromatic rectangle of density $2^{-\polylog(\r(M))}$.
\end{conjecture}

Sgall~\cite{sgall1999} considered \cref{conj:set_systems} in the case of 
$\{a,a+1\}$-cross-intersecting set systems and noted that it would follow from the log-rank conjecture. %We strengthen this by establishing a full equivalence between the two conjectures.
Sgall noted that any $\{a,b\}$-cross-intersecting set pair of families $(\mathcal{S}, \mathcal{T})$ from $2^{[d]}$ can be represented as an $m \times n$ matrix $M_{\mathcal{S}, \mathcal{T}}$ over $\{a,b\}$ with rank at most $d$, where $M_{\mathcal{S}, \mathcal{T}}(i,j) := |S_i \cap T_j|$. One can write $M_{\mathcal{S}, \mathcal{T}}$ as a sum of $d$ primitive matrices $R_k$ indexed by elements of $[d]$, where $R_k(i,j) = 1$ iff $k \in S_i \cap T_j$. Each $R_k$ has rank $1$, so the total rank of $M_{\mathcal{S}, \mathcal{T}}$ is at most $d$. Conversely, any $\{a,b\}$-valued matrix that is a sum of $d$ primitive matrices corresponds to an $\{a,b\}$-cross-intersecting set system over a universe of size $d$.

\setsystems*

\begin{proof}
	As mentioned above, the fact that \cref{conj:set_systems} is implied by the log-rank conjecture was proved by Sgall~\cite{sgall1999}. We include the proof here for completeness. For an $\{a,b\}$-cross-intersecting family pair $(\mathcal{S}, \mathcal{T})$ with $a \leq b$, we can write $M_{\mathcal{S}, \mathcal{T}} = (b-a) B + a J$, where $B$ is  some Boolean matrix and $J$ is the all-ones matrix. Then, $\r(B) - 1 \leq \r(M_{\mathcal{S}, \mathcal{T}}) \leq \r(B) + 1$. By \cref{conj:nw}, $B$ has a monochromatic rectangle of density $2^{-{\polylog}(\r(B))}$, which yields a monochromatic rectangle in $M_{\mathcal{S}, \mathcal{T}}$ of density $2^{-{\polylog}(\r(M))} \geq 2^{-\polylog(d)}$ from the discussion above.

	For the reverse direction, assume \cref{conj:set_systems} holds. Let $A$ be a Boolean $m \times n$ matrix with signed rectangle rank $u$. By \cref{thm:main} and \cref{conj:nw}, it suffices to find a monochromatic rectangle in $A$ of density at least $2^{-\polylog(u)}$.  We reduce this to the problem of finding a large $\{a\}$- or $\{b\}$-cross-intersecting subfamily in an $\{a,b\}$-cross-intersecting set system. Let $a,b$ be integers chosen later. Define a matrix $A' \in \{a,b\}^{m \times n}$ as follows:
	\[
	A'(i,j) := \begin{cases}
		a & \text{if } A(i,j) = 1, \\
		b & \text{if } A(i,j) = 0.
	\end{cases}
	\]
	
	Our goal is to show that $A'$ corresponds to an $\{a,b\}$-cross-intersecting set system over a universe of size $d = \Theta(u)$. 
	
	Let $A= \sum_{i=1}^u \eps_i R_i$ be the signed rectangle rank decomposition for $A$. We now construct a set of $2u$ primitive matrices $\{R'_i\}$ such that $A' = \sum_{i=1}^{2u} R'_i$, by replacing each $R_i$ with a pair of primitive matrices $R'_{2i-1}, R'_{2i}$ depending on the sign of $\eps_i$. 
	
	\begin{itemize}
		\item If $\eps_i = 1$, let $R'_{2i-1}=J$  (the all-ones matrix), and $R'_{2i} = R_i$. Then, $R'_{2i-1} + R'_{2i} = J + R_i$
		\item If \(\varepsilon_i = -1\), let \(A_i \subseteq [m]\), \(B_i \subseteq [n]\) be such that \(R_i = 1\) on \(A_i \times B_i\) and zero elsewhere.  
		Define \(R'_{2i-1}\) to be 1 on \(([m]\setminus A_i) \times [n]\) and 0 elsewhere, and \(R'_{2i}\) to be 1 on \(A_i \times ([n]\setminus B_i)\). Then \(R'_{2i-1} + R'_{2i} = J - R_i\). 
	\end{itemize}
	
	In either case, the pair \((R'_{2i-1}, R'_{2i})\) replaces \(\varepsilon_i R_i\) by \(J + \varepsilon_i R_i\). Thus, each entry of $A'$ satisfies $A'(i, j) = A(i, j) + u$, and therefore $A'$ corresponds to a $\{u, u+1\}$-cross-intersecting set family over $[d]$ for $d = 2u$. Applying~\cref{conj:set_systems}, we obtain a large monochromatic rectangle in $A'$, and hence in $A$, of density at least $2^{-\polylog(u)}$. 
\end{proof}

\begin{remark}
	The above proof shows that the log-rank conjecture is equivalent to a special case of~\cref{conj:set_systems} for $\{a, a+1\}$-cross-intersecting set systems.
\end{remark}

\section{Generalizing the Main Theorem to tensors}\label{sec:tensors}
The log-rank conjecture extends naturally to Boolean tensors. Motivated by this perspective, we show that \cref{thm:main} continues to hold for the tensor rank of constant-order Boolean tensors, giving an equivalent formulation of the log-rank conjecture for tensors.

Let $T$ be an order-$\ell$ Boolean tensor: a multilinear map $[n_1] \times \ldots \times [n_\ell] \to \{0, 1\}$ (for some natural numbers $n_1, \ldots, n_\ell$). In other words, $T$ can be expressed as an $\ell$-dimensional array whose entries are all in $\{0, 1\}$. 

We say $T$ has tensor rank $1$ if there are maps $v_i :[n_i] \to \R$ for $i \in [\ell]$ such that 
\[ T(x_1, \ldots, x_\ell) = v_1(x_1) \cdot \ldots \cdot v_{\ell}(x_{\ell}). \]
The \emph{tensor rank} of a tensor $T$, denoted by $\tr(T)$ is the minimum number $r$ such that $T$ can be expressed as the sum of $r$ rank-1 tensors. 

A \emph{primitive tensor} is the natural generalization of a primitive matrix: it is a rank-1 tensor where $v_i$'s map to $\{0, 1\}$. That is, a primitive tensor can be expressed as a multidimensional array which is all-1 on some product set $Q_1 \times \ldots \times Q_\ell$ for $Q_i \subseteq [n_i]$ and is all-0 elsewhere. Having this, the \emph{partitioning number} and \emph{signed rectangle rank} of a tensor are defined analogously as the minimum integers $p$ and $r$, respectively, such that $T=\sum_{i=1}^p S_i$ and $T=\sum_{i=1}^r \eps_{i} S_i$, where $S_i$ are primitive tensors and $\eps_i \in \{\pm1 \}$. 

The log-rank conjecture for tensors can be formulated as follows.
\begin{conjecture}[Log-rank conjecture for tensors]\label{conj:log-rank-tensors}
	For any constant-order Boolean tensor $T$, \[\log \p(T) \leq (\log \tr(T))^{O(1)}.\]
\end{conjecture}

In the communication complexity language, the log-rank conjecture for tensors states that the \emph{multiparty number-in-hand communication complexity} of a tensor $T$ is polynomially related to the logarithm of the tensor rank of $T$ (see \cite{DKWPartition2011}). It is known that the log-rank conjecture does not hold for tensors of super-polynomial order \cite{li2011applications}.

We prove that \cref{thm:main} generalizes to tensors.

\begin{theorem} \label{thm:tensors}
	Let $T: [n_1] \times \ldots \times [n_\ell] \to \{0, 1\}$ be an order-$\ell$ Boolean tensor with $\ell \geq 2$. If the tensor rank of $T$ is $r$, then $T$ can be written as a $\pm 1$-linear-combination of at most $(c r \log r)^{(\ell - 1)}$ primitive tensors, where $c$ is an absolute constant.
\end{theorem}

A direct corollary of this theorem is an equivalent formulation of the log-rank conjecture for tensors, analogous to \cref{conj:unary_vs_rank} for matrices.

\begin{conjecture}
	\label{conj:unary_vs_rank_tensors}
	Every constant-order Boolean tensor $T$ satisfies $\log \p(T) \leq (\log \ur(T))^{O(1)}$.
\end{conjecture}
%TODO: super-polynomial order, not true

Finally, let us set up the proof of \cref{thm:tensors}. A \emph{slice} of an order-$\ell$ tensor is the order-$(\ell-1)$ tensor obtained by taking a coordinate and setting it to a fixed value. For $\lambda \in [\ell]$, a \emph{$\lambda$-slice} is a slice where we specify that $\lambda$ is the coordinate to be fixed. Similar to the proof of \cref{thm:main}, we define a slice-sum to be the entrywise sum of slices, and an independent set of slices is one where all of its subsets have distinct slice-sums.

\begin{claim} \label{cl:IS-tensors}
	Let $S$ be an independent set of $\lambda$-slices of $T$. Then $|S| \leq O(r \log r)$, where $r$ is the tensor rank of $T$.
\end{claim}
\begin{proof}
	Let $m = \prod_{i \ne \lambda} n_i$.
	Consider the flattening of $T$ to a matrix $M_T \in \{0, 1\}^{m \times n_{\lambda}}$. Then $M_T$ has rank at most $r$ and, following the bijection between $\lambda$-slices of $T$ and columns of $M_T$, there is an independent set of columns $S'$ in $M_T$ with $|S'| = |S|$. Apply \cref{cl:IS-size} to obtain $|S'| \leq O(r \log r)$.
\end{proof}

\begin{claim}\label{cl:tensor_linear_combin}
	Let $S$ be a maximal independent set of $\lambda$-slices of $T$.  Then every $\lambda$-slice of $T$ can be expressed as a  $\pm 1$-linear combination of  $\lambda$-slices in $S$.
\end{claim}

We omit the proof of \cref{cl:tensor_linear_combin} as it is analogous to
\cref{cl:linear_combin}.

\begin{proof}[Proof of \cref{thm:tensors}]
	The proof is by induction. The base case is $\ell=2$, which is proven by \cref{thm:main}. For $\ell>2$, arbitrarily choose a coordinate  $\lambda$ and use \cref{cl:tensor_linear_combin} to express $T$ as a $\pm 1$-linear combination of a maximal independent set of $\lambda$-slices $S$. Each $\lambda$-slice of $S$ has tensor rank at most $r$, and so by the inductive hypothesis, each $\lambda$-slice in $S$ can be expressed as a $\pm 1$-linear combination of $(cr \log r)^{(\ell - 2)}$ order-$(\ell-1)$ primitive tensors.
	
	In a similar fashion to the proof of \cref{thm:main}, we therefore can write $T$ as a $\pm 1$-sum of ${|S| \cdot (cr \log r)^{(\ell - 2)}}$ order-$\ell$ primitive tensors. Conclude by substituting $|S| = O(r \log r)$ using \cref{cl:IS-tensors}.
\end{proof}

\begin{remark}
	We note that the proof of \Cref{thm:tensors} only uses the tensor rank as an upper bound for the \emph{flattening rank} of a tensor, which is the maximal rank of any flattening of it as an $n_\lambda \times \left( \prod_{i \ne \lambda} n_i \right)$ matrix over $\lambda \in [\ell]$. If this flattening rank is $r$, then the conclusion of \Cref{thm:tensors} still holds. 
	
	Moreover, in this formulation, the quantitative bound of $O((r \log r)^{\ell-1})$ primitive tensors is near optimal. Indeed, consider all tensors $T:[r]^{\ell} \to \{0,1\}$. Any such tensor has flattening rank $\le r$. On the other hand, a simple counting argument shows that most such tensors need $\Omega(r^{\ell-1}/\ell)$ primitive tensors in their decomposition.
	
	If we return however to the original formulation of \Cref{thm:tensors} using tensor rank, it is unclear if the bound is close to tight. In fact, we suspect that it can be significantly improved.
\end{remark}

\begin{question}
	Can the bound in \Cref{thm:tensors} be improved to $o(r^{\ell-1})$?
\end{question}

\section*{Acknowledgments} 
We are grateful to Adi Shraibman and Amir Yehudayoff for enlightening discussions.
\bibliographystyle{alpha}
\bibliography{ref}

\end{document}